\documentclass[preprint,12pt]{elsarticle}
\usepackage{amsmath, amssymb, amsthm, mathrsfs,amsfonts}
\usepackage{bbm}
\usepackage{caption}
\usepackage{subcaption}
\usepackage[bb=boondox]{mathalfa}
\usepackage{geometry}
\usepackage{graphicx}
\usepackage{hyperref}
\usepackage{enumitem} 
\usepackage{mathtools}
\usepackage{tikz}
\usepackage{dsfont}
\usetikzlibrary{graphs}
\usepackage{graphv1}

\DeclarePairedDelimiter{\floor}{\lfloor}{\rfloor}

\theoremstyle{plain}
\newtheorem{theorem}{Theorem}[section]
\newtheorem{lemma}[theorem]{Lemma}
\newtheorem{corollary}[theorem]{Corollary}

\newtheorem{definition}[theorem]{Definition}
\newtheorem{observation}[theorem]{Observation}
\newtheorem{claim}{Claim}

\theoremstyle{definition}

\theoremstyle{remark}
\newtheorem{remark}[theorem]{Remark}
\newtheorem{question}[theorem]{Question}

\newcommand{\R}{\mathbb{R}} 
\newcommand{\Z}{\mathbb{Z}} 
\def\gze {\Gamma_{E}(\mathbb{Z}_N)}
\def\gz {\Gamma(\mathbb{Z}_N)}
\def\gzi #1#2{\Gamma_E(\mathbb{Z}_{#1^{#2}})}
\def\gzj #1{\Gamma_E(\mathbb{Z}_{#1})}
\def\gr {\Gamma(R)}
\def\gre {\Gamma_{E}(R)}

\def\dimt {dim_{TH}}
\def\t #1{\text{#1}}

\begin{document}

\begin{frontmatter}
\title{The boxicity of the compressed zero divisor graph of the ring of integers modulo N}
\author{L. Sunil Chandran}
\affiliation{organization={Department of Computer Science and Automation, Indian Institute of Science},
            city={Bengaluru},
            postcode={560012}, 
            state={Karnataka},
            country={India}}
\author{
Suraj Kumar Sahoo}
\affiliation{organization={Department of Computer Science and Automation, Indian Institute of Science},
            city={Bengaluru},
            postcode={560012}, 
            state={Karnataka},
            country={India}}

\begin{abstract}

A $d$-dimensional box is the Cartesian product $[a_1,b_1]\times [a_2,b_2]\times  \cdots\times  [a_d,b_d]$ 
where each $[a_i,b_i]$ is a closed interval on the real line. The boxicity of a graph $G$, denoted by $box(G)$, is the minimum integer $d\geq 0$ such that $G$ is the intersection graph of a collection of $d$-dimensional boxes.

The class of zero divisor graphs introduced by Beck (1988) is a popular class of graphs associated with rings and has been studied extensively by several researchers. Suppose $Z(R)$ is the set of zero divisors of a ring $R$. The zero divisor graph $\Gamma(R)$ for a ring $R $ is defined as the graph with the vertex set $V(\Gamma(R))=Z(R)$ and $E(\Gamma(R))=\{\{x,y\}\colon x,y\in Z(R)\t{ with }x\neq y\text{ and }x y=0\}$. One can define an equivalence relation $\sim$ on $V(\Gamma(R))$ such that for distinct vertices $x$ and $y$, one has $x\sim y$ if and only if $x$ and $y$ have the same annihilator, i.e.,  $Ann(x)=Ann(y)$. The compressed zero divisor graph $\Gamma_E(R)$ for a ring $R$ is the simple graph obtained from $\Gamma(R)$ by retaining exactly one vertex from each equivalence class induced by $\sim$.

Despite being only (isomorphic to) an induced subgraph of a zero divisor graph, compressed zero divisor graphs retain a lot of structural information. Not only that, they are quite rich from a graph theoretic standpoint. For example, the disjointness graph of the non-empty proper subsets of a finite set of cardinality $t\geq 3$ is isomorphic to the compressed zero divisor graph of $\mathbb{Z}_N$ where $N$ is the product of $t$ many distinct prime numbers.

In this paper, we completely answer two open questions posed in Discrete Applied Mathematics 391 (2026), pp. 127-136. 
Let $N=\prod_{i=1}^a p_i^{n_i}$ be the prime factorization of a positive integer $N$ and let $\mathbb{Z}_N$ be the ring of integers modulo $N$. Our main contribution is determining the exact boxicity of the compressed zero divisor graph $\Gamma_E(\mathbb{Z}_N)$. We show that when $a\geq 2$, $box(\Gamma_E(\mathbb{Z}_N))= a-1$ if and only if one of the following is true: $(i)$ $a\geq 2$ and $N$ is the product of two coprime integers $x$ and $y$ such that $x$ is a square-free integer and $y$ is the cube of a prime number; $(ii)$ $a\geq 3$ and $N$ is square-free; $(iii)$ $a\geq 2$, $N$ is cube-free, not square-free, and contains at least one prime divisor $p_i$ such that $n_i=1$. If $a=2$ and $n_1=n_2=1$, then $\Gamma_{E}(\mathbb{Z}_N)$ is a clique, and so, $box(\Gamma_{E}(\mathbb{Z}_N))=0$. In all other cases, $box(\Gamma_{E}(\mathbb{Z}_N))=a$.

The proof for case $(i)$ uses the upper bound idea introduced in Discrete Applied Mathematics 391 (2026), pp. 127-136 in a more sophisticated way. The proof for case $(ii)$ is particularly interesting as it introduces fresh lower bounding ideas that go beyond the technique of merely finding an induced subgraph of large boxicity. Instead, it exploits the set of forbidden induced subgraphs of interval graphs hidden in $\Gamma_{E}(\mathbb{Z}_N)$ in a more targeted way. The proof for case $(iii)$ involves reducing the problem to a coloring problem on an auxiliary graph. 

\end{abstract}

\end{frontmatter}

\section{Introduction}
For a given family $\mathcal{F}$ of sets, the \textit{intersection graph} of $\mathcal{F}$ is the graph with $V(G)=\mathcal{F}$ where two vertices are adjacent if and only if the corresponding sets intersect. 
\begin{definition}
    A $d$-dimensional box (or $d$-box) is a set $[a_1,b_1]\times [a_2,b_2]\times \cdots \times [a_d,b_d]$ where $[a_i,b_i] $ are closed intervals in $\R$. The boxicity of a graph $G$ (denoted as $box(G)$) is the minimum integer $d\geq0$ such that $G$ can be represented as the intersection graph of $d$-boxes. 
\end{definition}
Interval graphs are precisely the graphs with boxicity at most 1. By convention, the boxicity of complete graphs is assumed to be 0.

Graphs that are representable as the intersection of certain `well-behaved' geometric objects in low dimensions are of considerable interest for mathematicians and computer scientists alike. This is because graphs of this kind usually enjoy several nice properties, and sometimes these properties can be exploited to get efficient algorithms for problems that are hard in general. For example, it is known that the graph isomorphism problem (this is the problem of determining whether two graphs given as input are isomorphic to each other) is polynomial time solvable for intersection graphs of unit-squares in the plane \cite{neuen:LIPIcs.ESA.2016.70}. For another example, the problem of finding the largest clique in a graph is efficiently solvable for geometric intersection graphs like unit-disk graphs \cite{raghavan2003robust} and co-comparability graphs \cite{grotschel1981ellipsoid}.
The intersection graphs of $d$-boxes also exhibit many nice properties. For example, it was noted by Chandran, Francis and Sivadasan~\cite{maxp}, that for graphs of order $n$ and boxicity at most $d$, the size of the largest clique can be determined in polynomial time. This is because these graphs contain at most $(2n)^d$ many maximal cliques and, therefore, one can find the largest clique by listing all such cliques in polynomial time (using the algorithm described in \cite{chiba1985arboricity}). Similarly, if one considers the problem of approximating the largest independent set, graphs of order $n$ and boxicity at most $d\geq 2$ admit an approximation factor of $(1+\frac{1}{c}\log_2 n)^{d-1}$ for any constant $c\geq 1$ (assuming that the axis-parallel box corresponding to each vertex of the graph is given as input)~\cite{agarwal1998label,berman2001efficient}. In contrast, for general graphs it is known that there is not even a $O(n^{1-\epsilon})$-factor approximation for the largest independent set, for any $\epsilon>0$, unless $P=NP$\cite{10.1145/1132516.1132612}. The boxicity of graphs also has applications in ecology~\cite{10.1007/BFb0070404} and operations research~\cite{opsut1981fleet}. Several bounds are known for boxicity in terms of parameters such as maximum degree~\cite{scott2020better}, tree width~\cite{chandran2007boxicity},vertex cover~\cite{chandran2009cubicity}, degeneracy and crossing number~\cite{ADIGA20142}, and bandwidth~\cite{chandran2013cubicity}. 

The study of boxicity in restricted graph classes is an interesting area of research. It is known that planar graphs have boxicity at most $3$~\cite{thomassen1986interval} and that outerplanar graphs have boxicity at most $2$~\cite{Schein}. Hartman, Newman and Ziv~\cite{hartman1991grid} showed that planar bipartite graphs have boxicity at most 2. The boxicity of line graphs, split graphs, asteroidal triple free graphs, circular arc graphs, strongly chordal graphs and chordal bipartite graphs have also been studied \cite{chandran2011boxicity,10.1007/978-3-319-12340-0_7,cozzens1983computing, bhowmick2010boxicity,bhowmick2011boxicity,leaf,chandran2011chordal}. See \cite{chandran2026survey} for a survey of results on the boxicity of graphs. 

In this paper, we study the boxicity of an algebraically defined class of graphs called \textit{the compressed zero divisor graph of a ring} defined as follows. One can define a simple graph on any commutative ring $R$ by taking its elements as vertices and adding an edge between two distinct vertices \( x \) and \( y \) if and only if $xy=0$ in $R$. This class of graphs was introduced by Beck~\cite{BECK1988208} and was later modified by Anderson and Livingston~\cite{ANDERSON1999434} who restricted its set of vertices to include only the zero divisors of $R$, i.e., the set $Z(R)=\{a\in R\setminus\{0\}\colon\exists b\in R\setminus\{0\} \t{  such that  } ab=0\}$. A graph defined this way is known as \textit{the zero divisor graph of $R$} and such graphs have been studied extensively (see, for example, \cite{9, ANDERSON20121626, 10,ANDERSON1999434,BECK1988208}). 
\begin{definition}
     The zero divisor graph $\gr$ for a ring $R$ is defined as the graph with vertex set $V(\gr)=Z(R)$ and edge set $E(\gr)=\{\{x,y\}\colon x,y\in Z(R)\t{ with }x\neq y\text{ and }x y=0\}.$
\end{definition}
 For any ring $R$, and for $x\in R\setminus\{0\}$ define $Ann_R(x)=\{y\in R\colon xy=0\}$. We say that two non-zero elements $x$ and $y$ are equivalent, denoted as $x\sim y$, if $Ann_R(x)=Ann_R(y)$. Observe that $\sim$ defines an equivalence relation on the vertices of $\Gamma(R)$. Let $Q=\{\langle x_1 \rangle, \langle x_2 \rangle, \ldots ,\langle x_k\rangle \}$ be the collection of equivalence classes of $\gr$ under the relation $\sim$. 
 \begin{definition}\label{compgr}
    The compressed zero divisor graph $\gre$ of a ring $R$ is defined as the graph with $V(\gre)=Q$ and $E(\gre)=\{\{\langle x\rangle, \langle y\rangle\}\t{ }\colon\t{ }\langle x\rangle \neq \langle y\rangle \t{ and } \{x,y\}\in E(\gr)\}$. 
\end{definition}
Even though the compressed zero divisor graph of a ring is only (isomorphic to) an induced subgraph of the zero divisor graph, it retains a lot of structural information about the ring. Moreover, these graphs are highly structured, and in some instances, they can have interesting combinatorial interpretations. For example, for a square-free positive integer $N=\prod_{i=1}^ap_i$ (here $p_i$ is a prime number for each $i\in [a]$), the graph $\gze$ is isomorphic to the disjointness graph \footnote{For a family $\mathcal{C}$ of sets, the disjointness graph of $\mathcal{C}$ is the complement of the intersection graph of $\mathcal{C}$.} of the collection of all non-empty proper subsets of $[a]$, i.e., the disjointness graph of $\mathcal{P}([a])\setminus \{\emptyset,[a]\}$ (here $\mathcal{P}([a])$ is the power set of $[a]$). Indeed, for $x\in V(\gze)$, the map $x\mapsto \{i\in [a]\colon p_i \t{ does not divide }x\}$ is the required isomorphism. Disjointness graphs of families of sets are well studied, a popular special case being when the family is the collection of all equal cardinality subsets of a finite set, that is exactly the class of Kneser graphs. Compressed zero divisor graphs of rings have been studied extensively (see, for example, \cite{mulay2002cycles,ANDERSON20121626,spiroff2011zero}).

  The study of the boxicity of zero divisor graphs of some special kinds of rings was initiated by Kavaskar~\cite{TK}. Suppose $N$ is a positive integer and $N=\prod_{i=1}^a p_i^{n_i}$ is the prime factorization of $N$ with primes $p_1<p_2< \cdots < p_a$, where $a\geq 2$ and $\forall i\in [a], n_i\geq 1$. Kavaskar~\cite{TK} showed that $box(\gz)\leq \prod_{i=1}^a(n_i+1)-\prod_{i=1}^a(\floor{n_i/2}+1)-1$. Later, $box(\gz)$ was exactly determined by Chandran and Sahoo~\cite{chandran2026boxicity} who proved the following theorem.
\begin{theorem}[\cite{chandran2026boxicity}]\label{thmgz}
      Suppose $N$ is a positive integer and $N=\prod_{i=1}^a p_i^{n_i}$ is the prime factorization of $N$ with primes $p_1<p_2< \cdots < p_a$, where $a\geq 2$ and $\forall i\in [a], n_i\geq 1$. Then, $$box(\Gamma(\mathbb{Z}_N))=\begin{cases}
        a-1 & \text{if } p_1= 2, n_1=1 \t{ and } n_i\leq 2 \text{ for }2\leq i\leq a\\
        a & \text{otherwise }
    \end{cases}$$
\end{theorem}
In the same paper, they posed the following question. 
\begin{question}\label{q1}
    What is the exact value of $box(\gze)$ and when does $box(\gze)=box(\gz)?$
\end{question}
A reduced ring $R$ is a ring for which $x\in R$, $x^2=0\implies x=0$. Let $R$ be a finite commutative reduced ring with identity. Suppose the chromatic number of $\gr$ is $k$. Kavaskar~\cite{TK} showed that $box(\gr)\leq 2^k-2$. This bound was improved by Chandran and Sahoo~\cite{chandran2026boxicity} by showing that $\floor{k/2}\leq box(\gr)\leq \dimt(\gr)\leq k$ (see Definition~\ref{thresholddim} for the definition of $\dimt$). In the same paper, they posed the following question.
\begin{question}\label{q2}
    Suppose $R$ is a non-zero finite commutative reduced ring with identity. Suppose the chromatic number of $\gr$ is $k\geq 3$ (the case $k\leq 2$ is trivial). Is the lower bound $box(\gr)\geq \lfloor\frac{k}{2}\rfloor$ tight? 
\end{question}
In this paper, we settle both Question~\ref{q1} and Question~\ref{q2} completely.  
\subsection{Our Results}\label{ourresults} Let $N$ be a positive integer and let $N=\prod_{i\in [a]} p_i^{n_i}$ be the prime factorization of $N$ with the primes $p_1<p_2< \cdots < p_a$ where $\forall i\in [a],$ $  n_i\geq 1$. The case where $a=1$ is easy because $\gze$ is empty or a clique for $n_1\leq 3$, and it is a threshold graph  for $n_1\geq 4$ (see Definition~\ref{thresholddef}). So we will assume $a\geq 2$. Let $P=\{i\in [a]\colon n_i=1\}$. We completely determine the boxicity of the compressed zero divisor graph of $\Z_N$. 

More precisely, we show the following theorem.
\begin{theorem}\label{thmmain}
   $box(\gze)=a-1$ if and only if one of the following is true, 
\begin{itemize}
    \item $a\geq 2$, $|P|=a-1$ and the unique index $j\in [a]\setminus P$ satisfies $n_j=3$.
    \item $a\geq 3$, $P=[a]$.
    \item $a\geq 2 $, $\emptyset\subsetneq P\subsetneq [a] $ and $\forall i\in [a]\setminus P$, $n_i=2$.
\end{itemize}
If $a=2$ and $n_1=n_2=1$, then $\gze$ is a clique, and so, $box(\gze)=0$. In all other cases with $a\geq 2$, $box(\gze)=a$.
\end{theorem}
The proof of Theorem~\ref{thmmain} also helps us to get the exact value of $\dimt(\gze)$ (see Definition~\ref{thresholddim} for the definition of $\dimt$). We have the following theorem.  
\begin{theorem}\label{thmthgze}
     $\dimt(\gze)=a-1$ if and only if $P=[a]$ and $a\geq 3$. If $\gze$ is a complete graph then $\dimt(\gze)=0$. In all other cases with $a\geq 2$, $\dimt(\gze)= a$.        
\end{theorem}
Further, Theorem~\ref{thmmain} has the following interesting consequence for the case where $N$ is square-free. 
\subsubsection*{Boxicity of the disjointness graph of the power set of $[a]$}
 As discussed earlier, disjointness graph of the non-empty proper subsets of $[a]$ is isomorphic to $\gze$ when $N$ is square-free. Therefore, boxicity of $\gze$ immediately tells us the boxicity of the disjointness graph of $\mathcal{P}([a])$. The boxicity of this graph is exactly $a-1$ when $a\geq 3$. This is because when $a\geq 3$, adding the sets $\emptyset$ and $[a]$ to the graph does not change the boxicity because in the new graph, the set $\emptyset$ is adjacent to every other set in the new graph and the set $[a]$ is adjacent only to $\emptyset$. When $a=1$, the disjointness graph of $\mathcal{P}([a])$ is an edge (which has boxicity $0$ because it is a clique) and when $a=2$, it is an interval graph. Therefore, the boxicity of the disjointness graph of $\mathcal{P}([a])$ is always $a-1$. In fact, the proof of Lemma~\ref{gzep} actually implies something stronger. Our proof implies that when $a\geq 3$, the disjointness graph induced by only the sets of cardinality 1 and 2 is already a graph with boxicity $a-1$. 

\begin{corollary}\label{pow}
    The disjointness graph of the power set $\mathcal{P}([a])$ has boxicity equal to $a-1$. Moreover, the subgraph induced by the sets $\binom{[a]}{1}\cup \binom{[a]}{2}$ has boxicity equal to $a-1$.
\end{corollary}

For positive integers $a$ and $t$, the Kneser graph $K(a,t)$ is defined to be the disjointness graph of the collection of all subsets of $[a]$ of cardinality $t$. It is known that for $a\geq 2t+1$, $box(K(a,t))\leq a-2$  \cite{caoduro2023boxicity} and $box(K(a,2))=a-2$ \cite{caoduro2023boxicityy}. 
Lemma~\ref{gzep} (and its proof) proves that when the singleton sets are also included in the graph $K(a,2)$ when $a\geq 5$, the boxicity increases exactly by 1.  

The boxicity of the disjointness graph of $\mathcal{P}([a])$ also answers Question~\ref{q2} completely. To see this, we first state a result by Anderson and LaGrange~\cite{ANDERSON20121626}. 
\begin{lemma}[\cite{ANDERSON20121626}]\label{R}
    Suppose $R$ is a reduced Noetherian ring with $1\neq 0$ and suppose the number of minimal prime ideals of $R$ is $k$. Let $R' = \mathbb{Z}_2^k = \mathbb{Z}_2 \times \mathbb{Z}_2 \times \cdots \times \mathbb{Z}_2$ ($k$-copies). Then $\gre\cong \Gamma(R')$. 
\end{lemma}
Suppose $R$ is a non-zero finite commutative reduced ring with identity and $k=\chi(\gr)$. Since finite rings are Noetherian and since $\chi(\gr)$ is equal to the number of minimal prime ideals of $R$ (see \cite{BECK1988208} for a proof), we have $\gre\cong \Gamma(\mathbb{Z}_2^k)$ by Lemma~\ref{R}.

Observe that there is a one-to-one mapping between the subsets $W\subseteq [a]$ and the indicator vectors $\mathbb{I}_W=(w_1,w_2,\ldots ,w_a)$, where for $i\in [a]$, $w_i=1$ if and only if $i\in W$. Note that $\mathbb{I}_W\in V(\Gamma(\mathbb{Z}_2^a))$ whenever $W\notin \{\emptyset,[a]\}$. Also, for two distinct subsets $W_1$ and $W_2$, we have $W_1\cap W_2=\emptyset$ if and only if $\mathbb{I}_{W_1}\cdot \mathbb{I}_{W_2}=\mathbb{0}$. Therefore, the disjointness graph of the non-empty proper subsets of $[a]$ is isomorphic to $\Gamma(\mathbb{Z}_2^a)$. This discussion along with Corollary~\ref{pow} gives the following corollary.
\begin{corollary}\label{cor:redringimprov}
    Let $R$ be a non-zero finite commutative reduced ring with identity and let $\chi(\gr)=k\geq 3$. Then $box(\gre)=k-1$. 
\end{corollary}
Since $\gre$ is an induced subgraph of $\gr$, by Observation~\ref{sub}, we have the following theorem that settles Question~\ref{q2}  (see Definition~\ref{thresholddim} for the definition of $\dimt$).
\begin{theorem}\label{thm:redringimprov}
    Let $R$ be a non-zero finite commutative reduced ring with identity and let $\chi(\gr)=k\geq 3$. Then $$k-1\leq box(\Gamma(R))\leq\dimt(\gr)\leq  k.$$ 
\end{theorem}

     The lower bound given in Theorem~\ref{thm:redringimprov} is tight. Let $p_1, p_2, \ldots ,p_k$ be distinct prime numbers not equal to 2. Take $ N=2\times \prod_{i=2}^{k} p_i$ and $R=\Z_N$. Note that $R$ is a reduced ring as every $u\in \gz$ has at least one prime factor missing and so $u^2\not\equiv 0\pmod N$. By Theorem~\ref{thmgz}, $box(\gr)=k-1$ and therefore, the lower bound is attained by this graph. On the other hand, if we take $N=\prod_{i=1}^{k} p_i$ and $R=\Z_N$, then Theorem~\ref{thmgz} gives $box(\gr)=k$. Therefore, Question~\ref{q2} is settled. 
     \subsection{Organization of the paper}
     The paper is organized as follows. First we lay down the essential background necessary for the paper in the preliminaries section. Then we settle a small corner case left open in \cite{chandran2026boxicity} about $\dimt(\gz)$. Then we move on to prove our main result, i.e., Theorem~\ref{thmmain}. We break down the proof into multiple parts for easy readability.  
     Let $N$ and $P$ be as defined in Section~\ref{ourresults}. We break down the cases as follows: 
     \begin{enumerate}
         \item $P=\emptyset$. This is handled in Lemma~\ref{cb}.
          \item $P\neq \emptyset$. We further break this case into the following subcases.
          \begin{enumerate}
           \item[2(a):] $P=[a]$. This case is handled in Lemma~\ref{gzep}.
           
            \item[2(b):] $\emptyset \subsetneq  P\subsetneq [a]$. We further subdivide this into more cases. 
            \begin{enumerate}
             \item[2(b)(i):] $\exists j\in [a]$ such that $n_j\geq 4$. This case is handled in Lemma~\ref{rest}.
             \item[2(b)(ii):] $\forall i\in [a],$ $ 1\leq n_i\leq 3$ and $\exists j,j'\in [a]$ such that $j\neq j'$ and $n_j=n_{j'}= 3$. This case is handled in Lemma~\ref{2bii}
             \item[2(b)(iii):] $\forall i\in [a],$ $ 1\leq n_i\leq 3$ and there is exactly one index $j\in [a]$ with $n_j= 3$,
         and $\exists l\in [a]\setminus \{j\}$ with $n_l=2$. This case is handled in Lemma~\ref{2biii}
          \item[2(b)(iv):] There is exactly one index $j\in [a]$ such that $n_j= 3$ and $P=[a]\setminus \{j\}$. This case is handled in Lemma~\ref{2biv}. 
          \item[2(b)(v):] $\forall i\in [a],$ $ 1\leq n_i\leq 2$. This case is handled in Lemma~\ref{2bv}. 
            \end{enumerate}
          \end{enumerate}
     \end{enumerate}
     This covers all the cases and proves Theorem~\ref{thmmain}. Lemmas~\ref{gzep},~\ref{2biii} and~\ref{2biv} are the most non-trivial among all lemmas in this paper and fully utilize the highly structured nature of $\gze$. In particular, the proof of Lemma~\ref{gzep} is especially interesting because it introduces fresh lower bounding ideas that go beyond the technique of simply finding an induced subgraph of large boxicity. Instead, it exploits the set of forbidden induced subgraphs of interval graphs hidden in $\Gamma_{E}(\mathbb{Z}_N)$ in a more targeted way. These three lemmas should be treated separately and considered the main contribution of this paper. 
     
     Finally, we prove Theorem~\ref{thmthgze} in a similar way by breaking it down into many cases. 
 \section{Preliminaries}

 All the graphs in this paper are finite and simple. We will denote the disjoint union of two sets with the symbol $\sqcup$. For a graph $G$, we denote by $\overline{G}$ the complement graph of $G$, i.e., the graph that has the same set of vertices $V(G)$ but with two distinct vertices adjacent if and only if they are not adjacent in $G$. For graphs $G_1$ and $G_2$ defined on disjoint sets of vertices, the join of $G_1$ and $G_2$, denoted as $G_1\vee G_2$, is the graph defined on the vertex set $V(G_1)\cup V(G_2)$ and with edge set $E(G_1)\cup E(G_2)\cup \{uv\colon u\in V(G_1), v\in V(G_2)\}$. 
 Given graphs $G$ and $H$, we say that graph $H$ is an \textit{induced subgraph} of $G$ if $V(H)\subseteq V(G)$ and $E(H)=\{uv\in E(G)\colon u,v\in V(H) \}$. Given a graph $G$ and a subset of vertices $A\subseteq V(G)$, \textit{the graph induced on $A$ in $G$} is the graph with $A$ as the set of vertices and the set $\{uv\in E(G)\colon u,v\in A \}$ as the set of edges. Given graphs $G$ and $H$, and a subset of vertices $A\subseteq V(G)$, we say that \textit{$A$ induces the graph $H$ in $G$} if the graph induced on $A$ in $G$ is isomorphic to $H$. For a positive integer $x$, we will use the notation $[x]$ to indicate the set $\{1,2,\ldots, x\}$. We will denote by $P_4$ the path on $4$ vertices, by $C_4$ the cycle on $4$ vertices and by $K_2$ the complete graph on two vertices, i.e., an edge. We will denote an interval graph and its interval representation by the same symbol when there is no confusion.

The following theorem is by Roberts~\cite{Roberts1969301} and will be used repeatedly throughout the paper (the statement in~\cite{Roberts1969301} is slightly different).
\begin{theorem}[\cite{Roberts1969301}]\label{Roberts}
    Suppose $G$ is a graph and let $K_G$ denote the complete graph defined on the vertices of $G$. Then the boxicity of $G$ is the minimum integer $t\geq0$ such that there exist interval graphs $\{I_i\}_{1\leq i\leq t}$ 
    with the property that for $1 \leq i\leq t$, $V(I_i)=V(G)$ and $E(G)=E(K_G)\cap E(I_1)\cap E(I_2)\cap \cdots \cap E(I_t)$. 
\end{theorem}
In Theorem~\ref{Roberts}, observe that when $G$ is the complete graph, the minimizing set of interval graphs is the empty set and hence $box(G)=t=0$. 

A sub-class of interval graphs called \textit{threshold graphs} will be of interest to us.
\begin{definition}[\cite{chvatal}]\label{thresholddef}
    A graph $G$ is said to be a threshold graph if there is a real number $S$ (the threshold) and a weight function $w\colon V(G)\rightarrow \mathbb{R}$ such that: $uv$ is an edge in $G$ if and only if $w(u) + w(v) \geq  S$. We denote this class of graphs by $TH$. 
\end{definition}
We also have the following equivalent characterization of the threshold graphs. 
\begin{theorem}[\cite{chvatal1973set}]
    A graph $G\in TH$ if and only if $G$ does not contain $P_4, C_4$ or $2K_2$ as an induced subgraph. Here $2K_2$ is the graph consisting of two disjoint edges.
\end{theorem}
If we replace ``interval graphs'' with ``threshold graphs'' in the statement of Theorem~\ref{Roberts}, we get the notion of \textit{the threshold dimension} of a graph. The representation of graphs as the intersection of ``special graphs'' (in the sense of Theorem~\ref{Roberts}) was studied by Kratochvíl and Tuza.~\cite{Kratochvil}. 
\begin{definition}\label{thresholddim}
    Let $G$ be a graph and let $K_{G}$ be the complete graph defined on the vertices of $G$. The threshold dimension of $G$ (denoted as $\dimt(G)$) is defined as the minimum integer $t\geq 0$ such that there exist threshold graphs $\{G_i\}_{1\leq i\leq t}$ with the following property: for $1\leq i\leq t$, $V(G_i)=V(G)$ and $E(G)=E(K_{G})\cap E(G_1)\cap E(G_2)\cap \cdots \cap E(G_t)$.
\end{definition}
As in Theorem~\ref{Roberts}, observe that when $G$ is the complete graph in Definition~\ref{thresholddim}, the minimizing set of threshold graphs is the empty set and hence $\dimt(G)=t=0$.
\begin{remark}
    In the literature, the phrase ``threshold dimension'' is sometimes used to denote the minimum number of threshold graphs required to \emph{cover} the edges of a graph. We will not use the phrase in this way but in the way of Definition~\ref{thresholddim} which is consistent with the notion of \textit{the intersection dimension} of graphs as defined by Kratochvíl and Tuza~\cite{Kratochvil}.
\end{remark}

Since threshold graphs are also interval graphs, the threshold dimension of a graph is an upper bound on its boxicity.
\begin{observation}\label{boxth}
    For any graph $G$, $box(G)\leq \dimt(G)$. 
\end{observation}
Removing vertices from a graph does not increase its boxicity (threshold dimension) because the same set of interval graphs (threshold graphs) without the removed vertices can be used. This fact is captured by the following observation.
\begin{observation}\label{sub}
    Let $G$ be a graph and $H$ be an induced subgraph of $G$. Then $box(G)\geq box(H)$ and $\dimt(G)\geq \dimt(H)$.
\end{observation}
We also have the following observation.
\begin{observation}\label{join}
    Suppose $G_1,G_2$ are two graphs with disjoint vertex sets. Then,
    \begin{itemize}
        \item $box(G_1\vee G_2)= box(G_1)+box(G_2)$ and,
        \item $\dimt(G_1\vee G_2)= \dimt(G_1)+\dimt(G_2)$.
        \end{itemize}
\end{observation}
Observations~\ref{sub} and ~\ref{join} can be exploited to get lower bounds for the boxicity of graphs. This is usually done by looking for induced subgraphs whose boxicity is well known or expressible as the join of two graphs. Such a graph, known as the ``Roberts' graph'', will be used later in our proofs.
\begin{definition}[Roberts' Graph]
    Let $G$ be the graph obtained by removing a perfect matching from a complete graph on $2t$ vertices. This graph is known as \emph{the Roberts' graph} on $2t$ vertices. We denote it by $\overline{tK_2}$ (the complement graph of $t$ many isolated edges).  
\end{definition}
The following observation is not hard to show and was used by Roberts in his pioneering paper to demonstrate the existence of graphs on $n$ vertices with boxicity $\lfloor\frac{n}{2}\rfloor$.
\begin{observation}\label{rob}
    Let $\overline{tK_2}$ be the Roberts' graph on $2t$ vertices. Then $\dimt(\overline{tK_2})\geq box(\overline{tK_2})=t$. 
\end{observation}
The following theorem serves as another lower bounding technique and was developed by Cozzens and Roberts~\cite{cozzens1983computing}.
  \begin{theorem}[\cite{cozzens1983computing}]\label{obs2}
   Let $G$ be a graph and $p\geq 1$. Suppose that $V(G)$ contains two disjoint sets of vertices $S_1=\{a_1, a_2, \ldots, a_{2p-1}\}$ and $S_2 = \{b_1, b_2,\ldots , b_{2p-1}\}$, and suppose that the only edges between $S_1$ and $S_2$ in $\overline{G}$ are $\{a_ib_i\colon i\in [2p-1]\}$. Then $box(G)\geq p$. 
\end{theorem}
\begin{corollary}\label{6graph}
    Let $G$ be a graph having $V(G)=\{v_1, v_2, \ldots, v_6\}$ such that:
    \begin{itemize}
        \item $v_1,v_2,v_3,v_4,v_5, v_6$ form a cycle in $G$, and
        \item $v_1v_4,v_2v_5,v_3v_6\notin E(G)$.
    \end{itemize}
    Then $box(G)>1$.
\end{corollary}
\begin{proof}[Proof of Corollary]
    In the statement of Theorem~\ref{obs2}, take $S_1=\{v_1,v_3,v_5\}$ and $S_2=\{v_4,v_6,v_2\}$. 
\end{proof}
\section{Results on $\dimt(\gz)$}
In this brief section we resolve a corner case left open in \cite{chandran2026boxicity}. They showed the following theorem.
\begin{theorem}[\cite{chandran2026boxicity}]\label{tho}
   Suppose $N=\prod_{i=1}^a p_i^{n_i}$, where $n_i>0$ and $a\geq 2$, be the prime factorization of $N$ with $p_1<p_2< \cdots < p_a$. Then,  \begin{enumerate}
         \item If $N\not\equiv2\pmod 4$ then $\dimt(\Gamma(\mathbb{Z}_N))=  a$.
         \item  If $N\equiv 2\pmod 4$ and $n_j\geq 3$ for some $2\leq j\leq a$ then $\dimt(\Gamma(\mathbb{Z}_N))=  a$.
         \item  If $N\equiv 2\pmod 4$ and $n_i\leq 2$ for $2\leq i\leq a$ then $a-1\leq \dimt(\Gamma(\mathbb{Z}_N))\leq  a.$
         \item If $N\equiv 2\pmod 4 \t{ and } n_i= 1 \text{ for }2\leq i\leq a$ then $\dimt(\gz)=a-1$.
     \end{enumerate} 
\end{theorem}
\noindent Here we exactly determine $\dimt(\gz)$ as shown below.
\begin{theorem}\label{th}
   Suppose $N=\prod_{i=1}^a p_i^{n_i}$, where $n_i>0$ and $a\geq 2$, be the prime factorization of $N$ with $p_1<p_2< \cdots < p_a$. Then,  \begin{enumerate}
         \item If $N\not\equiv2\pmod 4$ then $\dimt(\Gamma(\mathbb{Z}_N))=  a$.
         \item  If $N\equiv 2\pmod 4$ and $n_j\geq 2$ for some $2\leq j\leq a$ then $\dimt(\Gamma(\mathbb{Z}_N))=  a$.
         \item If $N\equiv 2\pmod 4 \t{ and } n_i= 1 \text{ for }2\leq i\leq a$ then $\dimt(\gz)=a-1$.
     \end{enumerate} 
\end{theorem}
\begin{proof}
    Parts (1) and (3) follow from Theorem~\ref{tho}. We prove part (2). 

    Note that in this case $p_1=2$ and $n_1=1$. Let $j\in [a]$ be the index such that $n_j\geq 2$. If $n_j\geq3$, then part (2) of Theorem~\ref{tho} gives $\dimt(\Gamma(\mathbb{Z}_N))= a$. Therefore, assume that $n_j=2$.
    
    Define $S=\left\{\frac{N}{2p_j},\frac{N}{p_j},\frac{N}{2},\frac{N}{p_j^2}\right\}$ and for $k\in [a]\setminus \{1,j\}$, define $S_k=\left\{\frac{N}{p_k^{n_k}},\frac{2N}{p_k^{n_k}}\right\}$. Observe that $S$ is an induced $P_4$ and $S_k$ is an induced independent set in $\gz$. Also, the join of these graphs exists as an induced subgraph of $\gz$. Therefore, by Observation~\ref{join} we have $$\dimt(\gz)\geq \dimt(S\vee \bigvee_{k\in [a]\setminus \{1,j\}} S_k)\geq 2+a-2=a.$$ The upper bound follows from Theorem~\ref{tho}. 
\end{proof}

\section{Proof of Theorem~\ref{thmmain}}
In this section, we exactly compute the boxicity and threshold dimension of the compressed zero divisor graph of the ring $\Z_N$, where $N$ is a positive natural number. 
\begin{observation}
    Let $N$ be a positive natural number. For distinct $x,y\in V(\gz)$, we have $x\sim y$ if and only if $\gcd(x,N)=\gcd(y,N)$. 
\end{observation}
\begin{observation}\label{obser}
   Let $D$ be the set of all divisors of $N$. Then the graph $\gze$ is isomorphic to the induced subgraph of $\gz$ defined on the vertex set $D\setminus \{1,N\}$. (Note that neither $N$ nor $1$ are zero divisors in $\Z_N$ and therefore, they must be excluded.)
\end{observation}
Due to Observation~\ref{obser}, we will simply consider a vertex of $\gze$ as an integer in $D\setminus\{1,N\}$. For two adjacent vertices $x,y\in V(\gze)$, we will sometimes abuse the notation and say $xy\in E(\gze)$ instead of saying $\{x,y\}\in E(\gze)$ whenever it is clear from the context what the intended meaning is. 
\begin{lemma}\label{cb}
Let $N=\prod_{i\in [a]} p_i^{n_i}$ be the prime factorization of $N$ with the primes $p_1<p_2< \cdots < p_a$ where $a\geq 2$ and for each $i\in [a]$, $n_i\geq 2$. Then $box(\gze)=a$.        
    
\end{lemma}

\begin{proof} Since $\gze$ is an induced subgraph of $\gz$, it is clear from Observation~\ref{sub} and Theorem~\ref{thmgz} that $box(\gze)\leq box(\gz)\leq a$. 
 
 We will show the lower bound. 
 For $i\in [a]$, take $S_i=\left\{\frac{N}{p_i^{n_i-1}},\frac{N}{p_i^{n_i}}\right\}$. Note that $\forall i\in [a]$, $\frac{N}{p_i^{n_i-1}}\times \frac{N}{p_i^{n_i}}$ does not contain $p_i^2$. Therefore, $\forall i\in [a], S_i$ is an independent set. It is also clear that for distinct $i,j\in [a]$, $S_i\cup S_j$ induces a $C_4$ in $\gze$. Therefore, the graph induced by $\bigcup_{i\in [a]}S_i$ is the Roberts' graph $\overline{aK_2}$ and hence by Observations~\ref{sub} and~\ref{rob} we have $box(\gze)\geq a$.    

\end{proof}
\begin{lemma}\label{gzep}
Let $N=\prod_{i\in [a]} p_i$ be the prime factorization of $N$ with the primes $p_1<p_2< \cdots < p_a$ where $a\geq 1$. If $a\leq 2$ then $box(\gze)=0$. If $a\geq 3$ then $box(\gze)=a-1$.
\end{lemma}
\begin{proof}
    When $a=1$, $\gze$ is the empty graph and when $a=2$, $\gze$ is a single edge. Therefore, $box(\gze)=0$.
    
    Now assume $a\geq 3$.
     We will first show the upper bound. Suppose that $2$ is a divisor of $N$. Since $\gze$ is an induced subgraph of $\gz$, by Theorem~\ref{thmgz} we have $box(\gze)\leq box(\gz)\leq a-1$ and we are done. Therefore, assume that $2$ is not a divisor of $N$. Let $j\in [a]$ be arbitrarily picked. Define $N'=2\times \prod_{i\in [a]\setminus \{j\}}p_i$. We will show that $\gze\cong\gzi{N'}{}$. The upper bound then follows by Theorem~\ref{thmgz} as $box(\gze)=box(\gzi{N'}{})\leq box(\Gamma( \Z_{N'}))\leq a-1$. By Observation~\ref{obser}, we know that any vertex of $\gze$ is of the form $\prod_{i\in [a]}p_i^{\alpha_i}$ where $\forall i\in [a], \alpha_i\in \{0,1\}$. Define the isomorphism $\phi\colon \gze\rightarrow \gzi{N'}{}$ as $\phi(\prod_{i\in [a]}p_i^{\alpha_i})=2^{\alpha_j}\times \prod_{i\in [a]\setminus \{j\}}p_i^{\alpha_i}$. It is easy to verify that $\phi$ is an isomorphism.
    
    Now we show the lower bound. When $N$ is the product of three primes, $\gze$ is the graph shown in Figure~\ref{figa} (take $p_{k_1}, p_{k_2}$ and $p_{k_3}$ to be the three prime divisors of $N$). This graph is not an interval graph because it contains an asteroidal triple. Therefore, $box(\gze)\geq 2=a-1$. So the lower bound holds when $a=3$.  
    
\begin{figure}[t!]
        \centering
        \begin{subfigure}[t]{0.4\textwidth}
        \resizebox{1\textwidth}{!}{
        \begin{tikzpicture}
\renewcommand{\defradius}{0.15}
\renewcommand{\vertexset}{(x,0,1.5,,,,,$\frac{N}{p_{k_1}}$,0.1cm,-90),(y,-1.35,-0.75,,,,,$\frac{N}{p_{k_2}}$,0.1cm,5),(z,1.35,-0.75,,,,,$\frac{N}{p_{k_3}}$,0.1cm,175),(a,-2.2242,-1.2375,,,,,$\frac{N}{p_{k_3}p_{k_1}}$,0.1cm,180),(b,2.2242,-1.2375,,,,,$\frac{N}{p_{k_2}p_{k_1}}$,0.1cm,0),(c,0,2.5,,,,,$\frac{N}{p_{k_3}p_{k_2}}$,0.1cm,0)}
\renewcommand{\edgeset}{(x,y),(x,z),(y,z),(a,y),(b,z),(x,c)}
\drawgraph
\end{tikzpicture} }
\subcaption{}
\label{figa}
\end{subfigure}
~
\hspace{20pt}
 \begin{subfigure}[t]{0.4\textwidth}
\resizebox{1\textwidth}{!}{
        \begin{tikzpicture}
        
\renewcommand{\defradius}{0.15}
\renewcommand{\vertexset}{(x,0,1.5,,,,,$\frac{N}{p_{k_1}}$,0.1cm,-90),(y,-1.35,-0.75,,,,,$\frac{N}{p_{k_2}}$,0.1cm,5),(z,1.35,-0.75,,,,,$\frac{N}{p_{k_3}}$,0.1cm,175),(a,-2.2242,-1.2375,,,,,$\frac{N}{p_{k_3}p_{k_1}}$,0.1cm,180),(b,2.2242,-1.2375,,,,,$\frac{N}{p_{k_2}p_{k_1}}$,0.1cm,0),(c,0,2.5,,,,,$\frac{N}{p_{k_3}p_{k_2}}$,0.1cm,0)}
\renewcommand{\edgeset}{(x,y),(x,z),(y,z),(a,y),(b,z),(x,c),(a,b)}
\drawgraph

\end{tikzpicture} }
\subcaption{}
\label{figb}
\end{subfigure}
~\\
\vspace{20pt}
 \begin{subfigure}[t]{0.4\textwidth}
\resizebox{1\textwidth}{!}{
        \begin{tikzpicture}
\renewcommand{\defradius}{0.15}
\renewcommand{\vertexset}{(x,0,1.5,,,,,$\frac{N}{p_{k_1}}$,0.1cm,-90),(y,-1.35,-0.75,,,,,$\frac{N}{p_{k_2}}$,0.1cm,5),(z,1.35,-0.75,,,,,$\frac{N}{p_{k_3}}$,0.1cm,175),(a,-2.2242,-1.2375,,,,,$\frac{N}{p_{k_3}p_{k_1}}$,0.1cm,180),(b,2.2242,-1.2375,,,,,$\frac{N}{p_{k_2}p_{k_1}}$,0.1cm,0),(c,0,2.5,,,,,$\frac{N}{p_{k_3}p_{k_2}}$,0.1cm,0)}
\renewcommand{\edgeset}{(x,y),(x,z),(y,z),(a,y),(b,z),(x,c),(a,b),(c,b)}
\drawgraph
\end{tikzpicture} }
\subcaption{}
\label{figc}
\end{subfigure}
~
\hspace{20pt}
 \begin{subfigure}[t]{0.4\textwidth}
\resizebox{1\textwidth}{!}{
        \begin{tikzpicture}
\renewcommand{\defradius}{0.15}
\renewcommand{\vertexset}{(x,0,1.5,,,,,$\frac{N}{p_{k_1}}$,0.1cm,-90),(y,-1.35,-0.75,,,,,$\frac{N}{p_{k_2}}$,0.1cm,5),(z,1.35,-0.75,,,,,$\frac{N}{p_{k_3}}$,0.1cm,175),(a,-2.2242,-1.2375,,,,,$\frac{N}{p_{k_3}p_{k_1}}$,0.1cm,180),(b,2.2242,-1.2375,,,,,$\frac{N}{p_{k_2}p_{k_1}}$,0.1cm,0),(c,0,2.5,,,,,$\frac{N}{p_{k_3}p_{k_2}}$,0.1cm,0)}
\renewcommand{\edgeset}{(x,y),(x,z),(y,z),(a,y),(b,z),(x,c),(a,b),(c,b),(a,c)}
\drawgraph
\end{tikzpicture} }
\subcaption{}
\label{figd}
\end{subfigure}
\caption{}
\label{fig}
    \end{figure}
We now assume that $a\geq 4$. 
For the sake of contradiction, suppose that the boxicity of $\gze\leq a-2$. Let $I_1, I_2, \ldots ,I_{a-2}$ be a box representation of $\gze$. Corresponding to every ordered pair $(j,k)\in [a]\times [a]$ with $j\neq k$, define the set $T(j,k)=\{N/p_jp_k$, $N/p_k\}$. Note that for distinct $j,k\in[a]$, the sets $T(j,k)$ and $T(k,j)$ represent two distinct pairs. Since $N$ is a square free integer, for distinct $j,k\in [a]$, $T(j,k)$ forms an independent set in $\gze$. For distinct $j,k\in [a]$, $T(j,k)$ represents an absent edge in $\gze$ and, therefore, it has to be absent in at least one of the interval graphs $I_1,I_2, \ldots ,I_{a-2}$. 
 Since there are $a(a-1)$ many ordered pairs of the form $(j,k)\in [a]\times [a]$ with $j\neq k$, $\exists l\in [a-2]$ such that $I_l$ simultaneously contains at least $\left\lceil\frac{a(a-1)}{a-2}\right\rceil=\left\lceil a\left(1+\frac{1}{a-2}\right)\right\rceil\geq  a+2$ many sets of the form $T(j,k)$ as absent edges. 

 Define \begin{align*}
    &S: =\{\{j,k\}\colon T(j,k) \t{ or } T(k,j) \t{ is an absent edge in }I_l\},\\&
S':=\{(j,k)\colon T(j,k)\t{ is an absent edge in }I_l \}.
\end{align*}
For a distinct pair $j,k\in [a]$, we will say that the set $\{j,k\}$ \emph{appears as a twin} in $I_l$, if both $T(j,k)$ and $T(k,j)$ appear in $I_l$ as absent edges. Let $W\subseteq S$ be the collection of all sets in $S$ that appear as twins in $I_l$. We have two cases:
\begin{enumerate}
    \item[Case 1:] Consider the case where $|W|\leq 2$. We prove the case when $|W|=2$ because the proof of the case $|W|\leq 1$ is similar. Let $W=\{\{j,k\},\{j',k'\}\}$ (here we allow for the possibility that $j=j'$ or $k=k'$). Assume without loss of generality that $j<k$ and $j'<k'$. Define $S''=S'\setminus \{(k,j),(k',j')\}$. Observe that because $|W|=2$, we have $|S|=|S''|\geq a$. By the Erd\H{o}s-Ko-Rado theorem, the largest intersecting family of 2-subsets of $[a]$ has at most $a-1$ elements. Therefore, there is a pair of disjoint sets $\{j_1,k_1\}$ and $\{j_2,k_2\}$ in $S$. Without loss of generality, assume that $T(j_1,k_1)$ and $T(j_2,k_2)$ are present as absent edges in $I_l$. Then $T(j_1,k_1)\cup T(j_2,k_2)$ forms an induced $C_4$ in $I_l$. This is a contradiction to the fact that $I_l$ is an interval graph. 
    \item[Case 2:] Consider the case where $|W|\geq 3$. Suppose $A_1,A_2,A_3\in W$ are distinct subsets. If there are distinct indices $i,i'\in \{1,2,3\}$ satisfying $A_i\cap A_{i'}=\emptyset$, then we have an induced $C_4$ as seen in Case 1. Therefore, assume that $A_1,A_2$ and $A_3$ are pairwise intersecting. There are exactly two patterns in which $A_1,A_2$ and $A_3$ can be pairwise intersecting.
    \begin{enumerate}
        \item[Case 2(a):] The first case occurs when there exist distinct $j,k_1,k_2,k_3\in [a]$ such that $A_1=\{j,k_1\}$, $A_2=\{j,k_2\}$ and $A_3=\{j,k_3\}$. In this case $T(j,k_1)\cup T(j,k_2)\cup T(j,k_3)$ induces a graph satisfying Corollary~\ref{6graph} in $I_l$. This is a contradiction.

        \item[Case 2(b):] The second case occurs when there exist distinct $k_1,k_2,k_3\in [a]$ such that $A_1=\{k_1,k_2\},$ $ A_2=\{k_2,k_3\}$ and $A_3=\{k_3,k_1\}$. In this case, the graph induced by $T(k_1,k_2)\cup T(k_2,k_3)\cup T(k_3,k_1)$ is one of the four graphs described in Figure~\ref{fig}. This is because for every distinct pairs $i,i'\in \{1,2,3\}$, both $T(k_i,k_{i'})$ and $T(k_{i'},k_i)$ are present as absent edges in $I_l$.  Note that the graph in Figure~\ref{figa} contains an asteroidal triple and the graphs in Figures~\ref{figb}, \ref{figc} and \ref{figd} all contain an induced $C_4$. Therefore, none of the graphs described in Figure~\ref{fig} are interval graphs. This contradicts the fact that $I_l$ is an interval graph.
    \end{enumerate}
\end{enumerate}
\end{proof}
\begin{lemma}\label{rest}
      Let $N=\prod_{i\in [a]} p_i^{n_i}$ be the prime factorization of $N$ with the primes $p_1<p_2< \cdots < p_a$ where $a\geq 1$ and for each $i\in [a]$, $n_i\geq 1$. Let $P=\{i\in [a]\colon n_i=1\}$.
      
           If $\emptyset \subsetneq P\subsetneq [a]$ and $\exists j\in [a]$ such that $n_j\geq 4$, then $box(\gze)= a$.

\end{lemma}
\begin{proof}
From Theorem~\ref{thmgz}, we know that $box(\gze)\leq box(\gz)\leq a$. Therefore, we will prove the lower bound.
   
 Define $S'_j=\left\{\frac{N}{p_j^{n_j-1}},\frac{N}{p_j^{n_j-2}}\right\}$. For $i\in [a]\setminus (P\cup \{j\})$, define $S_i=\left\{\frac{N}{p_i^{n_i}},\frac{N}{p_i^{n_i-1}}\right\}$ and for $k\in P$, define $S_k'=\left\{\frac{N}{p_k},\frac{N}{p_jp_k}\right\}$. It is easy to check that $S_j'$ is an independent set as $\frac{N}{p_j^{n_j-1}}\times \frac{N}{p_j^{n_j-2}}$ does not contain $p^{4}_j$. Similarly, for $k\in P$, the sets $S_k'$ are all independent sets as $\frac{N}{p_k}\times \frac{N}{p_jp_k}$ does not contain $p_k$. Moreover, it is easy to verify that for $k\in P\cup \{j\}$ and for $ i\in [a]\setminus (P\cup \{j\})$, the graph induced by $S_k'\cup S_i$ is the graph $\overline{2K_2}$. Similarly, for distinct $k,k'\in P\cup \{j\}$, the graph induced by $S'_k\cup S'_{k'}$ is  the graph $\overline{2K_2}$. It follows that we have the Roberts' graph $\overline{aK_2}$ as an induced subgraph, given by $$\bigcup_{k\in P\cup \{j\}}S_k'\cup \bigcup_{i\in [a]\setminus (P\cup \{j\})} S_i.$$ 
 Therefore, by Observations~\ref{sub} and~\ref{rob}, $box(\gze)\geq a$.
\end{proof}

\begin{lemma}\label{2bii}
     Let $N=\prod_{i\in [a]} p_i^{n_i}$ be the prime factorization of $N$ with the primes $p_1<p_2< \cdots < p_a$ where $a\geq 1$ and $\forall i\in [a],$ $ n_i\in [3]$. Let $P=\{i\in [a]\colon n_i=1\}$. 
    Suppose $N$ satisfies the following, 
     \begin{enumerate}[label=(\roman*)]
     \item $P\neq \emptyset$,
         \item $\exists j,j'\in [a]$ such that $j\neq j'$ and $n_j=n_{j'}= 3$.
     \end{enumerate}
      Then $box(\gze)= a$.
\end{lemma}
\begin{proof}
From Theorem~\ref{thmgz}, we know that $box(\gze)\leq box(\gz)\leq a$. Therefore, we will prove the lower bound.

 Suppose $\exists j,j'\in [a]$ such that $j\neq j'$ and $n_j=n_{j'}=3$. For $k\in P$, let $S'_k= \left\{\frac{N}{p_k},\frac{N}{p_jp_k}\right\}$. Define $S_j''=\left\{\frac{N}{p_j^2},\frac{N}{p_j^2p_{j'}}\right\}$ and $S_{j'}''=\left\{\frac{N}{p_{j'}^2},\frac{N}{p_{j'}^2 p_j}\right\}$. For $i\in [a]\setminus (P\cup\{j,j'\})$, define $S_i=\left\{\frac{N}{p_i^{n_i}},\frac{N}{p_i^{n_i-1}}\right\}$. Then the graph induced by $$S_j''\cup S_{j'}''\cup \bigcup_{k\in P}S_k'\cup \bigcup_{i\in [a]\setminus (P\cup \{j,j'\})} S_i$$ forms the Roberts' graph $\overline{aK_2}$ as an induced subgraph. Therefore, by Observations~\ref{sub} and~\ref{rob}, $box(\gze)\geq a$.
    \end{proof}
    \begin{definition}
    Suppose $G$ is a graph. We define the \emph{$C_4$ conflict graph} of $G$ (denoted as $\phi(G)$) as follows. The vertices of $\phi(G)$ are all the pairs of vertices $\{v_i,v_j\}\subset V(G)$ such that $v_i\neq v_j$ and $v_iv_j\notin E(G)$. Two vertices $\{v_i,v_j\}$ and $\{v_{i'},v_{j'}\}$ of $\phi(G)$ are adjacent if $\{v_i,v_j,v_{i'},v_{j'}\}$ induces a $C_4$ in $G$.
\end{definition}
\begin{observation}\label{conflict}
    For a graph $G$, $box(G)\geq \chi(\phi(G))$.
\end{observation}
\begin{proof}
     Any valid box representation of $G$ defines a proper coloring of $\phi(G)$. To see this, suppose $I_1, I_2, \ldots , I_d$ is a box representation of $G$. By definition, $\forall v_i,v_j\in V(G)$ if $v_iv_j\notin E(G)$, then $\exists m\in [d]$ such that $v_iv_j\notin E(I_m)$.  Then we color the vertex $\{v_i,v_j\}$ of $\phi(G)$ with the color $m\in [d]$ where $m$ is the smallest index for which $v_iv_j\notin E(I_m)$. It is clear that this is a valid coloring of $\phi(G)$ as two adjacent vertices are always assigned two different colors (otherwise there is an interval graph that contains an induced $C_4$).
\end{proof}
    \begin{lemma}\label{2biii}
     Let $N=\prod_{i\in [a]} p_i^{n_i}$ be the prime factorization of $N$ with the primes $p_1<p_2< \cdots < p_a$ where $a\geq 1$ and $\forall i\in [a],$ $ n_i\in [3]$. Let $P=\{i\in [a]\colon n_i=1\}$. 
     Suppose $N$ satisfies the following,
     \begin{enumerate}[label=(\roman*)]
         \item $P\neq \emptyset$,
         \item there is exactly one index $j\in [a]$ with $n_j= 3$,
         \item $\exists l\in [a]\setminus \{j\}$ with $n_l=2$.
     \end{enumerate} Then $box(\gze)= a$.
\end{lemma}
\begin{proof}
From Theorem~\ref{thmgz}, we know that $box(\gze)\leq box(\gz)\leq a$. Therefore, we will prove the lower bound.

    Consider the graph $\gzj{pq^2r^3 }$. 
    Consider the `$C_4$ conflict graph' $\phi(\gzj{pq^2r^3 })$ of $\gzj{pq^2r^3 }$. 
    
    Let $v_1\colon =\{qr^3,q^2r^3\},$ $ v_2\colon =\{pqr,pq^2r\},$ $v_3\colon =\{q^2r^2,q^2r^3\},$ $v_4\colon =\{pr^3,pqr^3\}$ and $v_5\colon =\{pq^2,pq^2r\}$. 
    
    It is easy to check that $v_1v_2v_3v_4v_5$ is a $5$-cycle in $\phi(\gzj{pq^2r^3 })$. Therefore, $\chi(\phi(\gzj{pq^2r^3 }))\geq 3$ as it contains an odd cycle. Therefore, $box(\gzj{pq^2r^3 })\geq 3$ by Observation~\ref{conflict}.   

    Now consider the case when $N=p_1^2p_2^3\prod_{i=3}^{a}p_i$. Consider the `$C_4 $ conflict graph' $\phi(\gze)$.  
    \begin{itemize}
        \item Let $v_1\colon=\left\{\frac{N}{p_2^2p_1},\frac{N}{p_2^2}\right\}$; $v_3\colon =\left\{\frac{N}{p_2^3},\frac{N}{p_2^2}\right\}$; $v_4\colon =\left\{\frac{N}{p_1^2},\frac{N}{p_1}\right\}$.
        \item For $k\in [a]\setminus \{1,2\}$, define $v_{2,k}\colon =\left\{\frac{N}{p_1p_k},\frac{N}{p_k}\right\}$ and $v_{5,k}\colon =\left\{\frac{N}{p_2p_k},\frac{N}{p_k}\right\}$ ;
    \end{itemize} 
    Recall that two vertices of $\phi(\gze)$ have an edge between them when their corresponding sets induce a $C_4$. Note that for $k\in [a]\setminus \{1,2\}$, the vertices $v_{2,k}$ define a clique of order $a-2$ in $\phi(\gze)$. Similarly for $k\in [a]\setminus \{1,2\}$, the vertices $v_{5,k}$ form a clique of order $a-2$ in $\phi(\gze)$. The vertex $v_3$ is adjacent to $v_4$. For $k\in [a]\setminus \{1,2\}$, the vertices $v_{2,k}$ are all adjacent to $v_1$ and $v_3$. For $k\in [a]\setminus \{1,2\}$, the vertices $v_{5,k}$ are all adjacent to $v_1$ and $v_4$.
    
    We claim that $\chi(\phi(\gze))\geq a$. Suppose for the sake of contradiction, we have $\chi(\phi(\gze))\leq a-1$. Let $v_1$ be colored with the color $1$. Then for $k\in [a]\setminus \{1,2\}$, the vertices $v_{2,k}$ use up all the colors from $\{2, 3, \ldots ,a-1\}$. The only available color for $v_3$ is $1$. In a similar way, it can be argued that $v_4$ must get color $1$. But $v_3$ and $v_4$ cannot get the same color as they are adjacent. This is a contradiction. Since $\chi(\phi(\gze))\geq a$, we must have $box(\gze)\geq a$ by Observation~\ref{conflict}. 

    Finally, suppose $N$ is as defined in the statement of the Lemma. Let $N'=\prod_{i\in P\cup \{j,l\}}p_i^{n_i}$ and let $N''=\frac{N}{N'}$. Let $A=\{u\in V(\gze)\colon N' \text{ divides }u\}$ and $B=\{u\in V(\gze)\colon N'' \text{ divides }u\}$. It is easy to check that the graph induced on the vertices of $A\setminus \{N'\}$ is isomorphic to $\gzi{N''}{}$ and the graph induced on the vertices of $B\setminus \{N''\}$ is isomorphic to $\gzi{N'}{}$. Moreover, the graph induced on the vertices of $(A\cup B)\setminus \{N',N''\}$ is the graph $\gzi{N''}{}\vee \gzi{N'}{}$. Note that $box(\gzi{N'}{})\geq |P|+2$ as shown above. 
    
    If $N''=1$, then $box(\gze)\geq box(\gzi{N'}{})\geq |P|+2=a$.
    
    If $N''$ has at least two prime divisors, by Lemma~\ref{cb}, $box(\gzi{N''}{})\geq a-|P|-2$. By Observations~\ref{sub} and ~\ref{join}, $box(\gze)\geq box(\gzi{N'}{}\vee \gzi{N''}{})= box(\gzi{N'}{})+box( \gzi{N''}{})\geq a$. 

    If $N''=q^2$ for some prime number $q$, then the graph induced on the vertices $\left\{\frac{N}{q},\frac{N}{q^2}\right\}\cup B\setminus \{N''\}$ is isomorphic to $\overline{K_2}\vee  \gzi{N'}{}$. Hence, by Observations~\ref{sub} and ~\ref{join}, $box(\gze)\geq box(\overline{K_2}\vee  \gzi{N'}{})\geq 1+|P|+2=a$.
\end{proof}
\begin{lemma}\label{2biv}
     Let $N=\prod_{i\in [a]} p_i^{n_i}$ be the prime factorization of $N$ with the primes $p_1<p_2< \cdots < p_a$ where $a\geq 2$ and $\forall i\in [a],$ $ n_i\in [3]$. Let $P=\{i\in [a]\colon n_i=1\}$.

     If there is exactly one index $j\in [a]$ such that $n_j= 3$ and $P=[a]\setminus \{j\}$, then $box(\gze)= a-1$.
\end{lemma}
\begin{proof}
    We first show the lower bound. 
     For $k\in P$ define $S_k=\left\{\frac{N}{p_k}, \frac{N}{p_kp_j}\right\}$. Then $\bigcup_{k\in P}S_k$ is the Roberts' graph on $2(a-1)$ vertices.
         
    Now we show the upper bound. Fix $l\in P$. For any prime number $q$ and any positive integer $x$, define the function $f(x,q)$ is the non-negative integer such that $q^{f(x,q)}$ divides $x$ but $q^{f(x,q)+1}$ does not divide $x$. 
    
    For $i\in P\setminus \{l\}, s\in \{0,1\}$ and $t\in \{0,1\}$, define the sets $$F_{i,s,t,0}:=\{u\in \gze\colon  f(u,p_i)=s,f(u,p_l)=t\t{ and }f(u,p_j)\in \{0,1\} \} ;$$ $$ F_{i,s,t,1}:=\{u\in \gze\colon f(u,p_i)=s,f(u,p_l)=t\t{ and }f(u,p_j)\in \{2,3\}\}.$$ 
    Additionally, we define two more sets, $$F_{i,0,0,(*)}:=\{u\in \gze\colon f(u,p_i)=0\t{ and }f(u,p_l)=0\};$$ $$F_{i,1,0,(*)}:=\{u\in \gze\colon f(u,p_i)=1\t{ and }f(u,p_l)=0\}.$$ 
    
    For each $i\in P\setminus \{l\}$, we create an interval graph $I_i$ as shown in Figure~\ref{f}. For the sake of completeness, we also include an explicit mapping of the vertices of $I_i$ to intervals on the real line as follows.
    $$I_i(v)=\begin{cases}
        [0,4.5]&\t{ if } v\in F_{i,1,1,1} \\
  [2,4.5] & \t{ if } v\in F_{i,1,0,(*)}\\
  [2+\frac{v}{2N},2+\frac{v}{2N}] & \t{ if } v\in F_{i,0,1,1}\\
  [\frac{v}{2N},\frac{v}{2N}] & \t{ if } v\in F_{i,0,0,(*)}\\
  [3.5+\frac{v}{2N},3.5+\frac{v}{2N}] & \t{ if } v\in F_{i,0,1,0}\\
[0,3] & \t{ if } v\in F_{i,1,1,0}.
    \end{cases}$$
    The following claim is easy to check by inspection.
    \begin{claim}\label{claim1}
        For $i\in P\setminus \{l\}$, if $f(u,p_i)+f(v,p_i)=0$, then $uv\notin E(I_i)$.
    \end{claim}

    For $t\in \{0,1\}$, $c\in \{0,1,2,3\}$, define the set $$J_{t,c}:=\{u\in \gze\colon  f(u,p_l)=t\t{ and }f(u,p_j)=c\} .$$

    We create an interval graph $I_j$ as shown in Figure~\ref{s}. For the sake of completeness, we also include an explicit mapping of the vertices of $I_j$ to intervals on the real line as follows.
    $$I_j(v)=\begin{cases}
    [0,6]& \t{ if }v\in J_{1,3}\\
  [4,6]& \t{ if }v\in J_{0,3}\\
  [1.5,5] & \t{ if }v \in J_{1,2}\\
 [2.75,5] &\t{ if }v\in J_{1,1}\\
 [1.5+\frac{v}{2N},1.5+\frac{v}{2N}] & \t{ if } v\in J_{0,1}\\
 [\frac{v}{2N},\frac{v}{2N}] &\t{ if } v\in J_{0,0}\\
   [5.25 + \frac{v}{2N}, 5.25 +\frac{v}{2N}] & \t{ if } v\in J_{1,0}\\ 
   [2.75+\frac{v}{2N},2.75 +\frac{v}{2N}]& \t{ if }v\in J_{0,2}\end{cases}$$
    
    The following claim can be checked by inspection. 
    \begin{claim}\label{claim2}
        $u,v\in V(I_j)$ are adjacent if and only if one of the following holds: $(i)$ $f(u,p_j)+f(v,p_j)\geq 3$ and $f(u,p_l)+f(v,p_l)\geq 1$, $(ii)$ $u,v\in J_{0,3}$ or $(iii)$ $u,v\in J_{1,1}$. 
    \end{claim} 

    We claim $E(\gze)=\bigcap_{k\in [a]\setminus\{l\}}E(I_k)$. It is easy to check by inspection that for $k\in [a]\setminus\{l\}$, $I_k$ is a supergraph of $\gze$. For distinct $u,v\in V(\gze)$, we will show that if $uv\notin E(\gze)$, then $uv\notin \bigcap_{k\in [a]\setminus \{l\}}E(I_k)$. To this end, suppose $uv\notin E(\gze)$. This means that $\exists y\in [a]$ such that $f(u,p_y)+f(v,p_y)<n_y$. 
    
    \begin{enumerate}
        \item[Case 1:] If $y=j$, then $uv\notin E(I_j)$ unless $u,v\in J_{1,1}$. Since $u$ and $v$ are distinct elements, $\exists z\in P\setminus \{l\}$ such that $f(u,p_z)\neq f(v,p_z)$. Without loss of generality, assume $f(u,p_z)=1$ and $f(v,p_z)=0$. Then $u\in F_{z,1,1,0}$ and $v\in F_{z,0,1,0}$. By construction, $uv\notin E(I_z)$. 
        \item[Case 2:] If $y=l$, then for $i\in P\setminus \{l\}$, we have $uv\notin E(I_i)$ unless $u,v\in F_{i,1,0,(*)}$. Suppose $\forall i\in  P\setminus \{l\}$, it is true that $u,v\in F_{i,1,0,(*)}$. Since $u$ and $v$ are distinct elements, it follows that $f(u,p_j)\neq f(v,p_j)$. As a consequence, it cannot happen that both $u$ and $v$ belong to $J_{0,3}$. This means that all the conditions of Claim~\ref{claim2} are violated. This means that $uv\notin E(I_j)$.
        \item [Case 3:] If $y\in P\setminus \{l\}$, then by Claim~\ref{claim1}, $uv\notin E(I_y)$.
    \end{enumerate}
   
       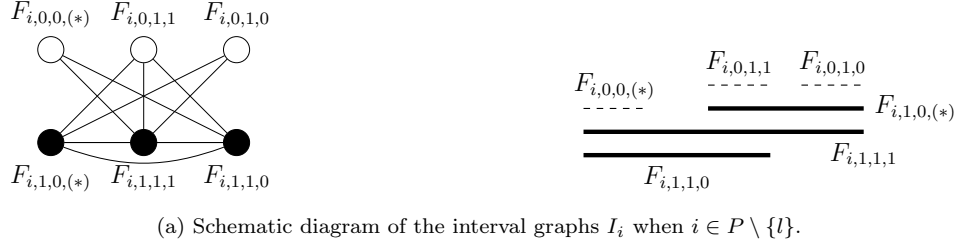
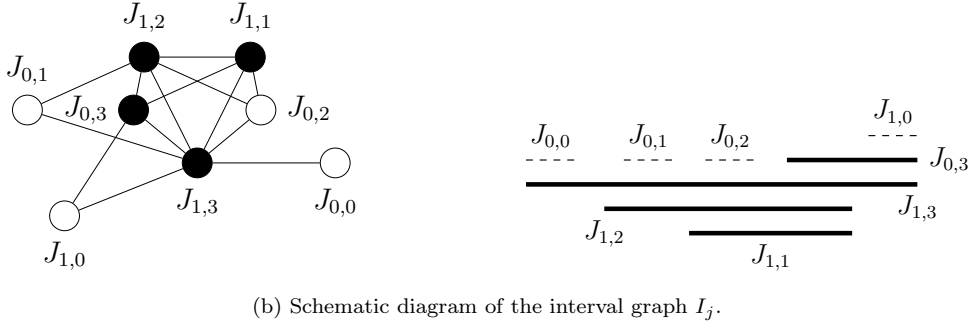
\begin{figure}[t!]
           \centering
       \begin{subfigure}[t]{1\textwidth}
       \centering
       \resizebox{0.25\textwidth}{!}{
            \begin{tikzpicture}[>=stealth,scale=1.5]
  \tikzstyle{filled} = [draw, circle, fill=black, minimum size=12pt, inner sep=0pt]
  \tikzstyle{unfilled} = [draw, circle, fill=white, minimum size=12pt, inner sep=0pt]

  \node[unfilled,label=above:{$F_{i,0,0,(*)}$}] (v00) at (0,1) {};
  \node[unfilled,label=above:{$F_{i,0,1,1}$}] (v01) at (1,1) {};
  \node[filled,label=below:{$F_{i,1,0,(*)}$}] (v10) at (0,0) {};
  \node[filled,label=below:{$F_{i,1,1,1}$}] (v11) at (1,0) {};
   \node[filled,label=below:{$F_{i,1,1,0}$}] (v111) at (2,0) {};
    \node[unfilled,label=above:{$F_{i,0,1,0}$}] (v011) at (2,1) {};

  \draw (v10) -- (v11);
  \draw (v10) -- (v01);
  \draw (v11) -- (v00);
  \draw (v11) -- (v01);
  \draw (v011)--(v10);
  \draw (v11)--(v011);
  \draw (v111) -- (v00);
  \draw (v111) -- (v01);
  \draw (v10) -- (v111);
  \draw (v111)--(v11);
  \draw (v10) to[bend right=20](v111);
\end{tikzpicture}}
 \hspace{100pt}
 \resizebox{0.35\textwidth}{!}{\begin{tikzpicture}[x=1cm,y=1cm]
\draw[line width=2pt] (6,3.125) -- (10.5,3.125) node[right,below]{$F_{i,1,1,1}$} ;

  \draw[line width=2pt] (8,3.5) --(10.5,3.5) node[right] {$F_{i,1,0,(*)}$};

  \draw[line width=0.5pt, dashed] (8,3.875)  -- (9,3.875) node[midway,above] {$F_{i,0,1,1}$};
  \draw[line width=0.5pt,dashed] (6,3.5)  -- (7,3.5) node[midway,above] {$F_{i,0,0,(*)}$};
     \draw[line width=0.5pt, dashed] (9.5,3.875)  -- (10.5,3.875) node[midway,above] {$F_{i,0,1,0}$};
   \draw[line width=2pt] (6,2.75) -- (9,2.75) node[midway,below]{$F_{i,1,1,0}$} ;
  \end{tikzpicture}}
  \subcaption{Schematic diagram of the interval graphs $I_i$ when $i\in P\setminus \{l\}$.}
  \label{f}
       \end{subfigure}
       \\[10pt]
\begin{subfigure}[t]{1\textwidth}
\centering
  \resizebox{0.325\textwidth}{!}{\begin{tikzpicture}[>=stealth,scale=1.5]
  \tikzstyle{filled} = [draw, circle, fill=black, minimum size=12pt, inner sep=0pt]
  \tikzstyle{unfilled} = [draw, circle, fill=white, minimum size=12pt, inner sep=0pt]

  \node[filled,label=above:{$J_{1,2}$}] (v21) at (0,1) {};
  \node[filled,label=above:{$J_{1,1}$}] (v11) at (1,1) {};
  \node[filled,label=left:{$J_{0,3}$}] (v30) at (-0.1,0.5) {};
  \node[unfilled,label=right:{$J_{0,2}$}] (v20) at (1.1,0.5) {};
  \node[filled,label=below:{$J_{1,3}$}] (v31) at (0.5,0) {};
  \node[unfilled,label=below:{$J_{0,0}$}] (v00) at (1.8,0) {};
  \node[unfilled,label=below:{$J_{1,0}$}] (v01) at (-0.75,-0.5) {};
  \node[unfilled,label=above:{$J_{0,1}$}] (v10) at (-1.1,0.5) {};

  \draw (v21) -- (v11);
  \draw (v21) -- (v30);
  \draw (v31) -- (v21);
  \draw (v31) -- (v11);
  \draw (v30)--(v11);
  \draw (v20)--(v21);
    \draw (v20) -- (v11);
  \draw (v31) -- (v30);
  \draw (v31) -- (v20);
  \draw (v31) -- (v00);
   \draw (v31) -- (v01);
    \draw (v31) -- (v10);
  \draw (v01)--(v30);
  \draw (v10)--(v21);
 
\end{tikzpicture}}
 \hspace{50pt}
 \resizebox{0.4\textwidth}{!}{\begin{tikzpicture}

 \draw[line width=2pt] (5,2.125) -- (11,2.125) node[below]{{$J_{1,3}$}};
  
  \draw[line width=2pt] (9,2.5) -- (11,2.5)node[right]{$J_{0,3}$};
  
  \draw[line width=2pt] (6.2,1.75)node[below]{$J_{1,2}$} -- (10,1.75);
 \draw[line width=2pt] (7.5,1.375) -- (10,1.375)node[midway,below]{$J_{1,1}$};
 \draw[line width=0.5pt,dashed] (6.5,2.5)--(7.25,2.5) node[midway,above]{$J_{0,1}$};
  \draw[line width=0.5pt,dashed] (5,2.5)  -- (5.75,2.5) node[midway,above]{$J_{0,0}$};
   \draw[line width=0.5pt,dashed] (10.25,2.875)  -- (11,2.875)node[midway,above]{$J_{1,0}$}; 
   \draw[line width=0.5pt,dashed] (7.75,2.5)  -- (8.5,2.5)node[midway,above]{$J_{0,2}$};

\end{tikzpicture}}
\subcaption{Schematic diagram of the interval graph $I_j$.}
\label{s}
\end{subfigure}

\caption{A schematic diagram of the graphs $I_i$ for $i\in P\setminus \{l\}$ (top), and the graph $I_j$ (bottom). Each node is labeled with a set and represents a group of vertices from that set. Solid nodes represent a clique and hollow nodes represent independent sets. Edges are present between two nodes if and only if every vertex in one node is adjacent to every vertex in the other node. The edges between the vertices of a solid node and the edges going between two vertices in two different nodes account for all the edges of these graphs. On the right, the corresponding intervals are illustrated for each vertex. The intervals corresponding to solid nodes are drawn long and thick to indicate that the vertices of a node are superposed on the same interval. The intervals corresponding to hollow nodes are drawn thin and narrow.}
\label{in}
\end{figure}
\end{proof}
\begin{lemma}\label{2bv}
      Let $N=\prod_{i\in [a]} p_i^{n_i}$ be the prime factorization of $N$ with the primes $p_1<p_2< \cdots < p_a$ where $a\geq 2$ and $\forall i\in [a],$ $ 1\leq n_i\leq 2$. If $\emptyset\subsetneq P\subsetneq [a]$, $box(\gze)=a-1$.
\end{lemma}
\begin{proof}
    When $N$ is not divisible by 2, the proof of the upper bound is similar to that of the proof of Lemma~\ref{gzep}. Suppose $N$ is divisible by $2$ (that is $p_1=2$) and $n_{\ell}=1$ for some $\ell\in [a]$ ($\ell=1$ is allowed). Let $f(\cdot,\cdot)$ be as defined in Lemma~\ref{2biv}. Define the following isomorphism \[N'=\frac{2\cdot p_{\ell}^{n_1}\cdot N}{2^{n_1}\cdot p_l}; \quad u\in \gze\mapsto \frac{2^{f(u,p_{\ell})}\cdot p_{\ell}^{f(u,2)}\cdot u}{2^{f(u,2)}\cdot p_l^{f(u,p_{\ell})}} \in \gzi{N'}{}.\] Therefore, by Theorem~\ref{thmgz}, $box(\gze)=box(\gzi{N'}{})\leq box(\Gamma(\Z_{N'}))=a-1$.
    
    We will now prove the lower bound. Choose $j\notin P$. For $k\in P$, put $
S_k=\left\{\frac{N}{p_kp_j},\frac{N}{p_k}\right\},$
and for $i\notin P\cup\{j\}$, put $
S_i'=\left\{\frac{N}{p_i^2},\frac{N}{p_i}\right\}.$ Then, $$\bigcup _{k\in P}S_k\cup \bigcup_{i\notin P\cup \{j\}}S_i' $$ forms a Roberts' graph on $2(a-1)$ vertices. Therefore, $box(\gze)\geq a-1$.
     
\end{proof}
\section{Proof of Theorem~\ref{thmthgze}}
\begin{proof}[Proof of Theorem~\ref{thmthgze}]
We show the following four claims.
\begin{enumerate}
        \item If $P=\emptyset$, then we have $\dimt(\gze)=a$.
        \item If $P=[a]$ and $a\geq 3$ then $\dimt(\gze)=a-1$.
        \item If $P=[a]$ and $a\leq 2$ then $\dimt(\gze)=0$.
       
        \item If $\emptyset \subsetneq P\subsetneq [a]$ then we have $ \dimt(\gze)=a$.
    \end{enumerate}
    Claim 1 follows from Lemma~\ref{cb}, Theorem~\ref{th} and observing that $box(\gze)\leq \dimt(\gze)\leq \dimt(\gz)$. 
    
    To see Claim 2, take $N'=2\times \prod_{i=2}^ap_i$ and observe that there is an isomorphism between $\gze$ and $\gzi{N'}{}$ (see the proof of Lemma~\ref{gzep}). Then the claim follows from Theorem~\ref{th},Lemma~\ref{gzep} and observing that $ box(\gze)\leq \dimt(\gze)=\dimt(\gzi{N'}{})\leq \dimt(\Gamma(\Z_{N'}))\leq a-1$.
    
    Claim 3 follows trivially as in this case $\gze $ is either empty or a single edge, and therefore a complete graph. 
    
    We will show claim 4. Since $ P\subsetneq [a]$, there exists an index $j\in [a]\setminus P$. Define $S_j=\left\{\frac{N}{p_j^{n_j}},\frac{N}{p_j^{n_j-1}}\right\}$. For $k\in P$, define the sets $S_k=\left\{\frac{N}{p_kp_j},\frac{N}{p_k}\right\}$ and $\forall i\in [a]\setminus (P\cup \{j\})$, define $S_i=\left\{\frac{N}{p_i^{n_i}},\frac{N}{p_i^{n_i-1}}\right\}$. It is easy to check that for $r\in [a]$, $S_r$ is an independent set. Also, for distinct $r,r'\in [a]$, $S_r\cup S_{r'}$ induces a $C_4,P_4$ or a $2K_2$. Moreover, for any supergraph that contains $S_r$ and $S_{r'}$ as non-edges, the graph induced by $S_r\cup S_{r'}$ is a $C_4,P_4$ or a $2K_2$. Therefore, in any threshold super graph of $\gze$, at most one of the sets in the collection $\{S_r\}_{r\in [a]}$ can exist as an independent set. Since each of these sets must exist as an independent set in some threshold graph, we must have that $\dimt(\gze)\geq a$. By Theorem~\ref{th}, $\dimt(\gze)\leq a$. 
\end{proof}
\bibliographystyle{plain}
\bibliography{cas-refs}
\end{document}